%% file: imperf_BOX_LS_ver01.tex
\documentclass[journal,11pt,onecolumn,draftclsnofoot,]{IEEEtran}

\IEEEoverridecommandlockouts
\usepackage{cite}
\usepackage{amsmath,amssymb,amsfonts}
\usepackage{algorithmic}
\usepackage{graphicx}
\usepackage{textcomp}
\usepackage{xcolor}
\usepackage{amsthm}
\usepackage{caption}
\usepackage{relsize}
\usepackage{subcaption}
\usepackage{float}

\def\BibTeX{{\rm B\kern-.05em{\sc i\kern-.025em b}\kern-.08em
    T\kern-.1667em\lower.7ex\hbox{E}\kern-.125emX}}
    
\usepackage{systeme}
\usepackage{color}
\usepackage{pgfplots}
\usepackage{bbm}
    \DeclareMathOperator{\tr}{tr}
  \newcommand{\figref}[1]{Fig.~\protect\ref{#1}}

\newcommand{\bg}{{\bf g}}

 	\newcommand{\bh}{{\bf h}}

           \newcommand{\bI}{{\bf I}}

                            \newcommand{\pto}{\overset{P}\longrightarrow }

\include{macros}

\newtheorem{theorem}{Theorem}

\newtheorem{remark}{Remark}
\newtheorem{lemma}{Lemma}
\newtheorem{assumption}{Assumption}

\begin{document}
\makeatother
\def\x{{\mathbf x}}
\def\L{{\cal L}}

\title{Large System Analysis of Box-Relaxation in Correlated Massive MIMO Systems Under Imperfect CSI (Extended Version)
\thanks{
The author is with the Department of Electrical Engineering, College of Engineering, University of Ha'il, P.O. Box 2440, Ha'il, 81441, Saudi Arabia (e-mail:
am.alrashdi@uoh.edu.sa). This work was supported by the University of Ha'il, Saudi Arabia.
%
%
}
}
\author{\IEEEauthorblockN{Ayed M. Alrashdi\\}
\IEEEauthorblockA{
Email: am.alrashdi@uoh.edu.sa
}
}
\maketitle
\begin{abstract}
In this paper, we study the mean square error (MSE) and the bit error rate (BER) performance of the box-relaxation decoder in massive multiple-input-multiple-output (MIMO) systems under the assumptions of imperfect channel state information (CSI) and receive-side channel correlation. Our analysis assumes that the number of transmit and receive antennas ($n,\text{and} \ m$) grow simultaneously large while their ratio remains fixed. For simplicity of the analysis, we consider binary phase shift keying (BPSK) modulated signals. The asymptotic approximations of the MSE and BER enable us to derive the optimal power allocation scheme under MSE/BER minimization. Numerical simulations suggest that the asymptotic approximations are accurate even for small $n$ and $m$. They also show the important role of the box constraint in mitigating the so called double descent phenomenon.
\end{abstract}
\section{Introduction}
\label{sec:intro}
Massive multiple-input multiple-output (MIMO) is considered one of the key enabling technologies for the future cellular
networks as it promises significant gains in data rates, spectral and energy efficiency, and link reliability \cite{marzetta2010noncooperative, hoydis2013massive, larsson2014massive, lu2014overview, ngo2013energy, hossain2014evolution}. The idea of massive MIMO is to use a large number of antennas at the base station to serve many users at the same time and frequency resources.
To attain the significant benefits of massive MIMO, accurate knowledge of the channel state information (CSI) is required. A popular approach for acquiring CSI is through training by sending a known sequence of pilot symbols. Thus, prior knowledge of the transmitted pilot sequence is directly incorporated in the process of estimating the CSI. Following this step, the receiver employs the estimated CSI to detect the corresponding transmitted data symbols.

The quality of the recovered symbols can be improved by controlling the power allocation between the transmitted pilot sequences and data symbols to meet specific optimization criteria. Different metrics have been proposed for the power allocation optimization starting from maximizing the channel capacity \cite{gottumukkala2009, kannu2005capacity, hassibi2003much}, maximizing the sum rates \cite{dao2018pilot, zhu2018uplink, lu2018training}, minimizing the mean square error (MSE) \cite{ballal2019optimum, zhao2017game} or minimization of the bit error rate (BER) and symbol error rate (SER) \cite{wang2014ber, alrashdi2020optimum, alrashdi2018optimum}, and many other metrics depending on the specific application. 

The power allocation in the aforementioned works was considered mainly for uncorrelated channel models. In practice, wireless communication systems, including massive MIMO systems, are generally spatially correlated \cite{bjornson2017massive}. The power optimization problem was developed for correlated channels to maximize the sum rates \cite{wagner2012large, muharar2020optimal}, or the spectral efficiency \cite{boukhedimi2018lmmse, cheng2016optimal}.

In this work, we propose the use of the box-relaxation optimization (BRO) \cite{tan2001constrained, yener2002cdma, thrampoulidis2016ber} as a low complexity decoder for a spatially correlated massive MIMO system. First, we derive novel precise asymptotic approximations of its MSE and BER performance using binary phase shift keying (BPSK) signaling for simplicity. Then, these approximations are used to derive the optimal power allocation scheme. The essential technical tool used in our analysis is the recently developed convex Gaussian min-max theorem (CGMT) \cite{thrampoulidis2018precise, thrampoulidis2018symbol}.

The CGMT framework has been used to analyze the error performance of many problems under independent and identically distributed (iid) assumption on the entries of the channel matrix \cite{thrampoulidis2018precise, thrampoulidis2018symbol, alrashdi2017precise, atitallah2017box, alrashdi2019precise, dhifallah2020precise, hayakawa2020asymptotic}. Furthermore, for correlated channel matrices, the CGMT was recently used in \cite{alrashdi2020box, alrashdi2020precise} to analyze the BRO and the LASSO decoders, respectively. However, these references assume the ideal scenario of perfect knowledge of the CSI matrix.
At the same time, this work tackles the more practical and challenging case of imperfect CSI under correlated channel matrix.

The remainder of this paper is organized as follows.
Section \ref{sec:system_mod} describes the considered system model. The asymptotic analysis and optimal power allocation are presented in Section~\ref{sec:main}. Section~\ref{sec:simu} presents the numerical simulations that show the high accuracy of the proposed theoretical results. A proof outline of the main results is given in Section~\ref{sec:proof}. Finally, the paper is concluded in Section~\ref{sec:conclusion}.
\section{System Model}\label{sec:system_mod}
We consider a flat block-fading massive MIMO system with $n$ transmitters and $m$ receivers. The channel between the transmit and receive antennas is modeled as \cite{adhikary2013joint, alfano2004capacity}
\begin{align}\label{eq:MIMO}
\Am = \Rm^{1/2} \Hm,
\end{align}
where $\Rm \in \mathbb{R}^{m \times m}$ is known Hermitian positive semi-definite spatial correlation matrix which models the correlation among the receive side antennas,\footnote{For analysis purpose, we assume that $\Rm$ satisfies $\frac{1}{m} \text{tr}(\Rm) = \mathcal{O}(1)$.} while $\Hm \in \mathbb{R}^{m \times n}$ is a Gaussian matrix with iid entries $\mathcal{N}(0,1)$ that represents Rayleigh fast-fading channel. Using this block fading model, each channel coherence interval of length $T$ is split into two phases starting by training and followed with data
transmission.
In the training interval of $T_p \geq n$ symbols, orthogonal sequences of known pilot symbols with average power $\rho_p$ are transmitted which allows to estimate the channel. The remaining phase is devoted for transmitting $T_d= T - T_p$ data symbols with average power $\rho_d$.
 
Conservation of time and energy implies that:
$
\rho_p T_p + \rho_d T_d = \rho T,
$
where $\rho$ is the expected average power.
Alternatively, we can write $ \rho_d T_d= \alpha \rho T$, where $\alpha \in [0,1]$ is the ratio of the power allocated to the data, and then
$\rho_p T_p =  (1- \alpha) \rho T$ is the energy of the pilots.
The received signal $\yv \in \mathbb{R}^m$ model for the \emph{data} transmission phase is given by
\begin{equation}
\yv = \sqrt{\frac{\rho_d}{n}} \Am \xv_0 + \zv,
\end{equation}
where $\xv_0 \in \{\pm1 \}^n$ is a BPSK signal, $\Am \in \mathbb{R}^{m \times n}$ is the MIMO channel given in \eqref{eq:MIMO}, and $\zv \in \mathbb{R}^m$ is the noise vector that has iid entries $\mathcal{N}(0,1)$. 
\subsection{Channel Estimation}
The channel matrix $\Am$ needs to be estimated prior to decoding the received data signal. Letting $\widehat{\Am}$  to denote the estimate of the channel matrix,
in this work, we consider linear minimum mean square error (LMMSE) estimate which is given by \cite{kay1993fundamentals}
\begin{align}
\widehat{\Am} = \sqrt{\frac{n}{\rho_p}} \Rm \left(\Rm + \frac{n}{T_p \rho_p} \Id_m \right)^{-1} \Ym_p \Xm_p^\top,
\end{align}
where $\Ym_p =  \sqrt{\frac{\rho_p}{n}} \Am \Xm_p + \Zm_p \in \mathbb{R}^{m \times T_p}$ is the received signal corresponding to the \emph{training} phase, $\Xm_p \in \mathbb{R}^{n \times T_p}$ is the matrix of transmitted orthogonal pilot symbols, and $\Zm_p \in \mathbb{R}^{m \times T_p}$ is an additive white Gaussian noise (AWGN) matrix with $\mathbb{E} [\Zm_p \Zm_p^\top]= T_p \Id_m$.
According to \cite{kay1993fundamentals, boukhedimi2018lmmse}, the $k^{th}$ column (for all $k\leq n$) of $\widehat{\Am}$ is distributed as
$\mathcal{N}({\bf{0}},\Rm_{\widehat{\Am}})$
with a covariance matrix $\Rm_{\widehat{\Am}}$ that is given by
$
\Rm_{\widehat{A}} = \Rm \left(\Rm + \frac{n}{T_p \rho_p} \Id_m \right)^{-1} \Rm.
$

Note that the pilots energy $T_p \rho_p$ controls the quality of the estimation. In fact, as $T_p \rho_p \to \infty$, $\widehat{\Am} \to \Am$ which corresponds to the perfect CSI case.
By invoking the orthogonality principle of the LMMSE estimator, it can be shown that the $k^{th}$ column of the estimation error matrix $\Deltam = \widehat{\Am} - \Am$ follows the distribution $\mathcal{N}({\bf{0}},\Rm_{\Delta})$ with a covariance matrix
$
\mathbf{R}_{\Delta} = \Rm - \Rm_{\small{\widehat{A}}}.
$
Also, from the orthogonality principle of the LMMSE, one can show that $\widehat{\Am}$ and $\Deltam$ are statistically independent.
\subsection{Symbol Detection via BRO}
Since modern massive MIMO systems are equipped with a large number of antennas, the computational complexity of optimum algorithms, such as the maximum-likelihood (ML) and sphere decoders, increases exponentially as the problem size grows. Many heuristic low-complexity detection algorithms are thus
proposed such as zero-forcing (ZF), linear minimum mean square error (LMMSE) and lattice reduction.
To obtain a reasonable computational complexity, one popular heuristic that is used in this paper is the so called box-relaxation optimization (BRO) decoder \cite{tan2001constrained, yener2002cdma}, which is a natural convex relaxation of the optimum ML decoder. In this decoder, the discrete set $\{\pm1\}^n$ is relaxed to the convex set $[-1,+1]^n$, and now the signal can be recovered via efficient convex optimization followed by hard thresholding. The BRO decoder has been shown to outperform conventional decoders such ZF decoder as we will explain in the next sections.

The BRO decoder consists of two steps. The first step involves solving a convex quadratic optimization with linear constraints. The output of this optimization is hard-thresholded in the second step to produce the desired binary solution. Formally, the algorithm produces an
estimate $\xv^*$ of $\xv_0$ given as
\begin{subequations}\label{eq:BRO}
\begin{align}
&\widehat{\xv}: = \argmin_{-1\leq \xv \leq 1} \frac{1}{n} \| \yv - \sqrt{\frac{\rho_d}{n}} \widehat\Am \xv \|^2 ,\label{step1} \\
&{\xv}^{*}: = {\rm{sign}}(\widehat{\xv})\label{step2},
\end{align}
\end{subequations}
where $\| \cdot \|$ denotes the $\ell_2$-norm of a vector and the $\mathrm{sign}(\cdot)$ function returns the
sign of its input and acts element-wise on vector inputs.

We consider the following two metrics to evaluate the performance of the BRO decoder:\\
\textit{Mean Square Error}: This metric is concerned with the evaluation of performance of the first step (estimation step) of the decoder in \eqref{step1}. Formally, the mean square error (MSE) is defined as
\begin{equation}
\mathrm{MSE} := \frac{1}{n}\| \widehat{\xv}  - \xv_0\|^2.
\end{equation}
\textit{Bit Error Rate}: For the second step (detection step) in \eqref{step2}, we evaluate the performance of the decoder by the bit error rate (BER), defined as
\begin{equation}
{\rm{BER}} := \frac{1}{n} \sum_{i=1}^{n} \mathbbm{1}_{\{{x}^*_i \neq x_{0,i} \}},
\end{equation}
where $\mathbbm{1}_{\{\cdot\}}$ denotes the indicator function.
%
\section{Large System Analysis and Power Optimization}\label{sec:main}
\subsection{Large System Error Analysis}
In this section, we present the large system analysis of the BRO decoder in \eqref{eq:BRO} in terms of the MSE and BER. Then, these results will be used later to find the optimal power allocation. Before doing so, we need to state some technical assumptions first.
\begin{assumption}\label{Assump1}
We assume that the number of transmit and receive antennas $n$ and $m$ are growing large to infinity with a fixed ratio $\frac{m}{n} \to \beta$, for a fixed constant $\beta > \frac{1}{2}$.\footnote{The phase transition at $\frac{1}{2}$ is the threshold for the BRO to successfully recovers the true signal $\xv_0$. \cite{thrampoulidis2018symbol}}
\end{assumption}
\begin{assumption}\label{Assump2}
We assume that the normalized coherence time, normalized number of pilot symbols and normalized number data symbols are fixed and given as 
$$
\frac{T}{n} \to \tau \in (1,\infty),
$$
$$
\frac{T_p}{n} \to \tau_p \in [1,\infty),
$$
and 
$$
\frac{T_d}{n} \to \tau_d,
$$
respectively.
\end{assumption}
Note that under Assumption \ref{Assump2}, the covariance matrix of $\widehat{\Am}$ becomes
$
\Rm_{\widehat{A}} = \Rm \left(\Rm + \frac{1}{\tau_p \rho_p} \Id_m \right)^{-1} \Rm.
$
 Define the spectral decomposition of $\Rm_{\widehat{A}}$ as $\Rm_{\widehat{A}} = \Um \Lambdam \Um^\top$, where $\Um \in \mathbb{R}^{m \times m}$ is an orthonormal matrix and $\Lambdam \in \mathbb{R}^{m \times m}$ is a diagonal matrix with the eigenvalues of $\Rm_{\widehat{A}}$ on its main diagonal. Let $Q(\cdot)$ denote the $Q$-function associated with the standard normal probability density function (pdf).
Finally, for a sequence of random variables $\{ \varTheta_n\}_{[n=1,2,\cdots]}$, we write $\varTheta_n \pto \varTheta$ to denote convergence in probability towards a constant $\varTheta$ as $n \to \infty$.
\begin{theorem}[MSE of the BRO] 
Let $\rm MSE$ denote the mean square error of the BRO decoder in \eqref{eq:BRO}, then under Assumption \ref{Assump1} and Assumption \ref{Assump2}, it holds
\begin{equation}\label{eq:mse}
\left| \mathrm{MSE} -F(\mu_*) \right| \pto 0,
\end{equation}
where $\mu_*$ is the unique solution to the following scalar optimization problem:
\begin{equation}\label{cost_th}
\min_{\mu>0} \max_{\gamma>0} \frac{1}{2 n} \sum_{j=1}^{m} \frac{\rho_d \lambda_{j} F(\mu) + \rho_d [\Rm_\Delta]_{jj}+1}{\frac{1}{2} + \frac{\lambda_j \sqrt{\rho_d}}{\gamma}} -\frac{\sqrt{\rho_d}}{2} \Upsilon^2(\mu) \gamma ,
\end{equation}
and $F(\mu) :=4 \big( Q(\mu) + \frac{\varphi(\mu)}{\mu^2}\big), \Upsilon(\mu):=\frac{1}{\mu} \big(1- 2 Q(\mu) \big)$, $\varphi(\mu):=\frac{1}{2}- Q(\mu) -\frac{\mu}{\sqrt{2 \pi}} e^{\frac{-\mu^2}{2}}$ and $\lambda_j$ is the $j^{th}$ eigenvalue of the matrix $\Rm_{\widehat{A}}$.
\end{theorem}
\begin{proof}
The proof is deferred to Section \ref{sec:proof}.
\end{proof}
\begin{remark}\normalfont
Note that the above MSE result holds for $\xv_0$ drawn from any distribution with zero mean and unit variance and not necessarily from a BPSK constellation.
\end{remark}
The next theorem presents the asymptotic BER approximation of the BRO decoder.
\begin{theorem}[BER of the BRO] 
Let $\rm BER$ denote the bit error rate of the BRO decoder in \eqref{eq:BRO}. Then, under the same setting of Theorem 1, it holds that
\begin{equation}\label{eq:ber}
\left| \mathrm{BER} - Q \bigg(\frac{\mu_*}{2} \bigg) \right| \pto 0.
\end{equation}
\label{thmmm2}
\end{theorem}
\begin{proof}
The proof is given in Section \ref{sec:proof}.
\end{proof}
\begin{remark} \normalfont
Although the CGMT requires an asymptotic regime in which $m, n \to \infty$, the approximations are already accurate for small values of $m, n$, for example see \figref{fig:ber_b}.
\end{remark}
\subsection{Optimal Power Allocation}
In this subsection, we will use the previous asymptotic approximations of the MSE and BER to find the optimum power allocation between pilot and data symbols to minimize the MSE or BER. For a fixed $\tau_p$ and $\tau$, the power allocation optimization can be cast as
\begin{align*}
\alpha_*^{\rm{MSE}} := \argmin_{0\leq\alpha\leq1} \rm{MSE},
\end{align*}
where $\rm {MSE}$ is the asymptotic MSE expression in \eqref{eq:mse}. Similarily, we have 
\begin{align*}
\alpha_*^{\rm{BER}} := \argmin_{0\leq\alpha\leq1} \rm{BER}.
\end{align*}
However, based on \eqref{eq:ber}, since minimizing the $Q$-function amounts to maximizing its argument, we have
\begin{align*}
\alpha_*^{\rm{BER}} := \argmax_{0\leq \alpha \leq 1} \mu_*.
\end{align*}
\figref{fig:power} illustrates the optimized data power ratio $\alpha_*$ versus the correlation coefficient $r$ for a total average power of $\rho= 10\ \rm dB$, where we used the exponential correlation model for $\Rm$ which is defined as \cite{shin2006capacity, loyka2001channel} 
\begin{equation}\label{exp_corr}
\Rm(r) =  r^{| i-j|^2}, r \in [0,1), i,j=1,2,\cdots,m.
\end{equation}
From this figure, we can see that $\alpha_*^{\rm{MSE}}=\alpha_*^{\rm{BER}}$. This indicates that optimizing the MSE is equivalent to optimizing the BER for the considered BRO decoder. Similar observation was found in \cite{alrashdi2020optimum}. In \figref{fig:power2}, we find the power allocation for low total average SNR of $\rho= -10\ \rm dB$.

As another illustration, in \figref{fig:power3}, we plot the asymptotic approximations of the MSE and BER in Theorem~1 and Theorem~2 respectively as a function of the data power ration $\alpha$. From both figures, we can see that $\alpha_* \approx 0.55$ for $r \in [0,0.9]$.
\begin{figure}
\begin{center}
\includegraphics[width = 10.8cm]{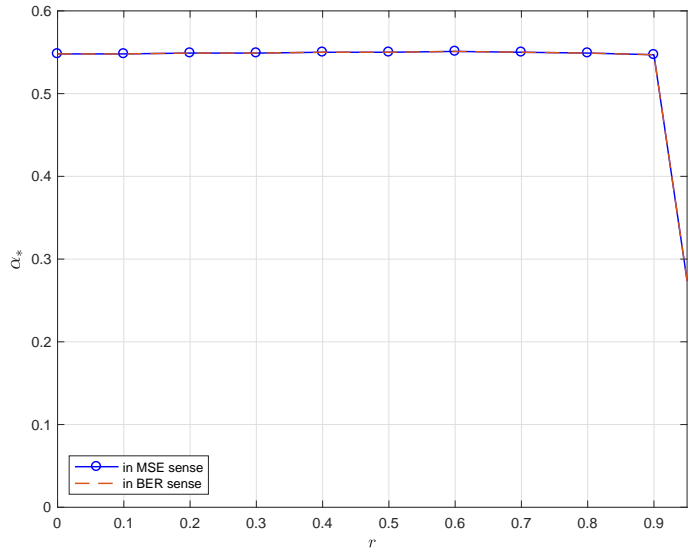}
\end{center}
\caption{\scriptsize {Optimal data power ratio $\alpha_*$ v.s. correlation coefficient $r$, with $\beta =1.5, n=500,\rho=10 {\rm {\ dB}}, T=1000, T_p =n$.}}%
\label{fig:power}
\end{figure}
\begin{figure}
\begin{center}
\includegraphics[width = 10.8cm]{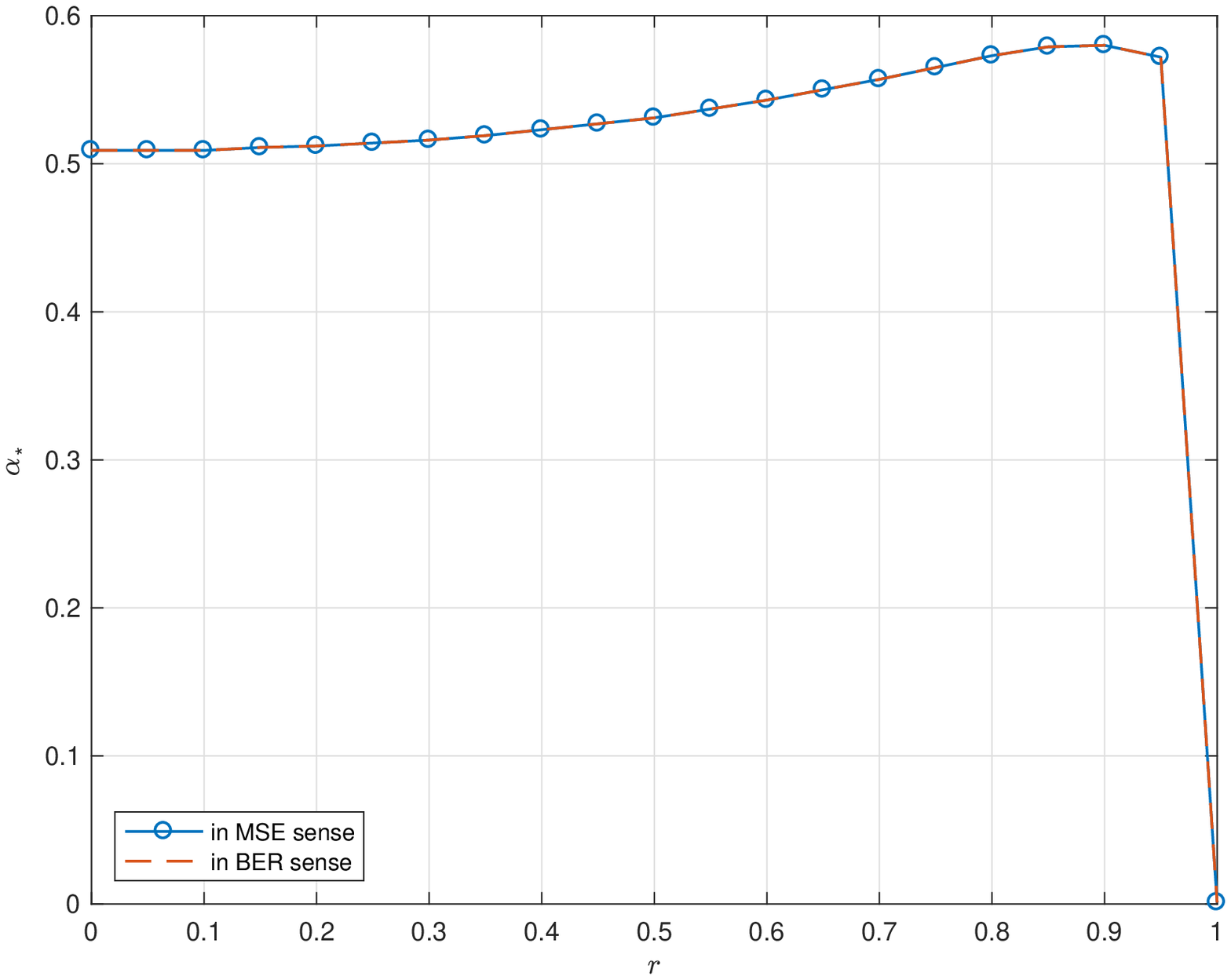}
\end{center}
\caption{\scriptsize {Optimal data power ratio $\alpha_*$ v.s. correlation coefficient $r$, with $\beta =1.5, n=500,\rho=-10 {\rm {\ dB}}, T=1000, T_p =n$.}}%
\label{fig:power2}
\end{figure}
\begin{figure*}
  \centering
\begin{subfigure}[h]{.5\textwidth}
\includegraphics[width =8.8cm]{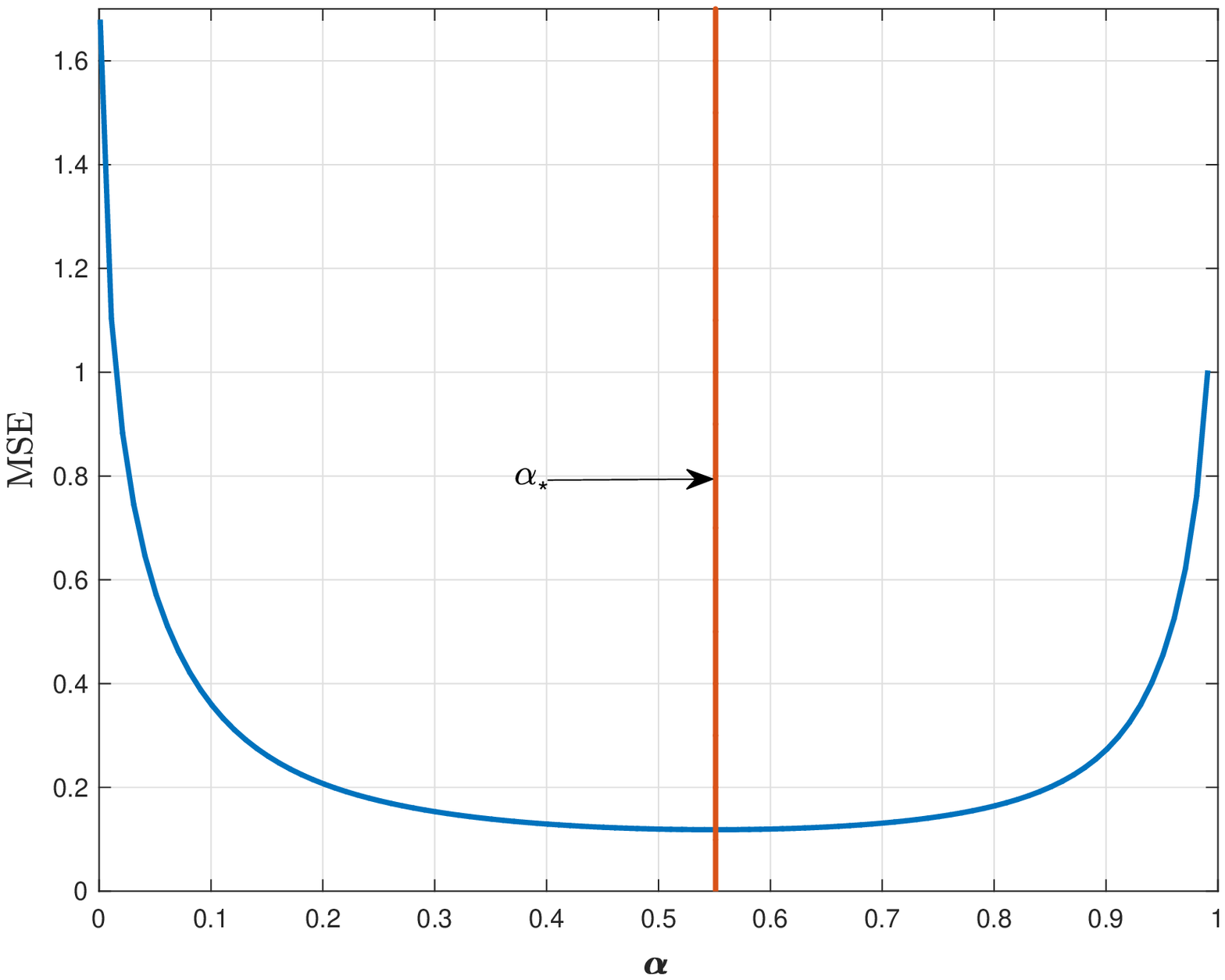}
\caption{\scriptsize {MSE performance.}}%
  \label{fig:sub-first}
\end{subfigure}%
\begin{subfigure}[h]{.5\textwidth}
\includegraphics[width =8.8cm]{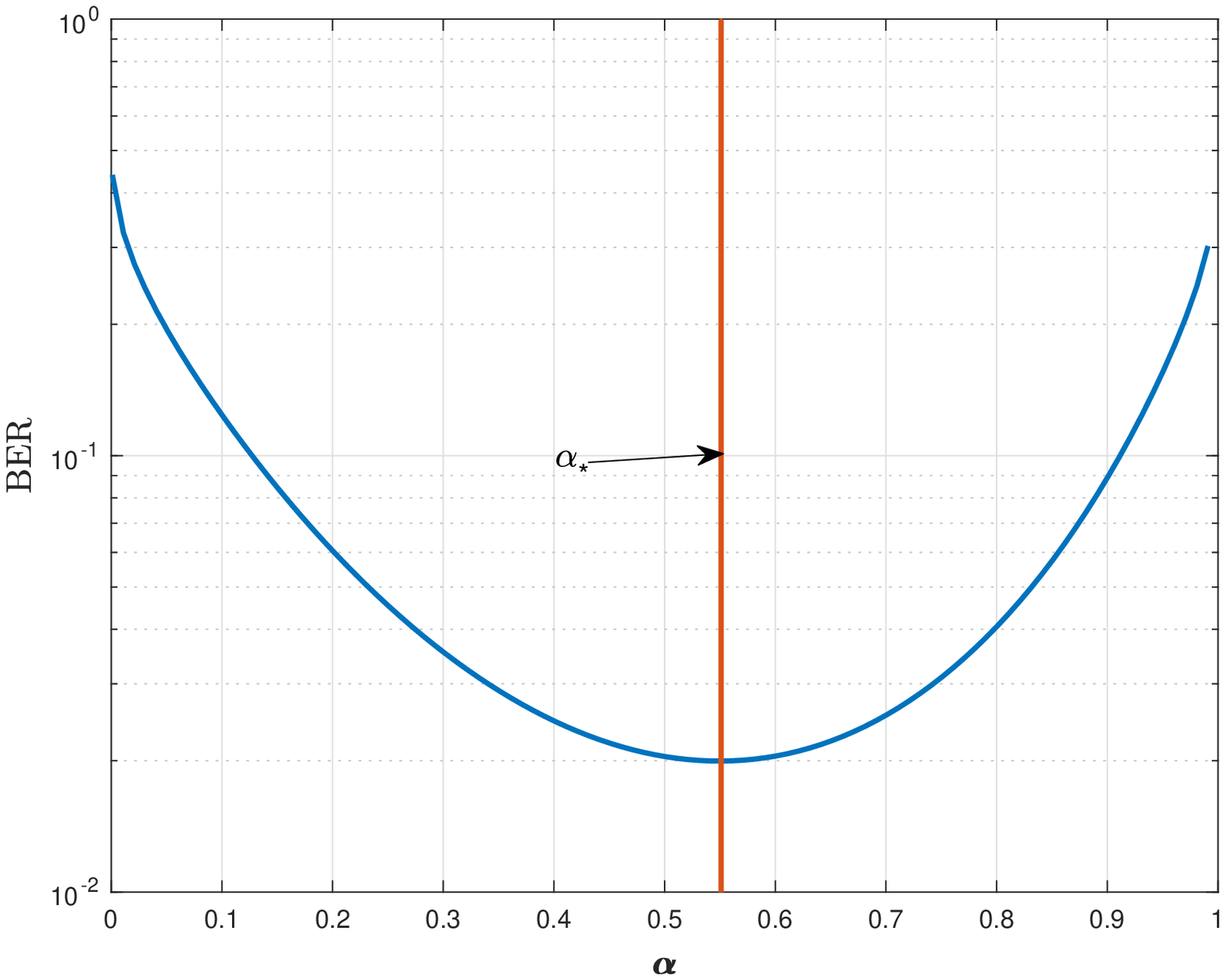}
  \label{fig:sub-second}
  \caption{\scriptsize {BER Performance.}}
\end{subfigure}
\caption{\scriptsize {Performance of BRO decoder vs $\alpha$, with $\beta =1.5, n=500, r=0.4, \rho=10 {\rm {\ dB}}, T=1000, T_p =n$.}}
\label{fig:power3}
\end{figure*}
\section{Numerical Results}\label{sec:simu}
To validate our theoretical predictions of the MSE as given by Theorem 1 and BER as stated in Theorem 2, 
we consider the exponential correlation model for $\Rm$ defined in \eqref{exp_corr}.
\figref{fig:mse_b1} shows the MSE performance of the BRO decoder for different values of the average expected power $\rho$ and different values of the correlation coefficient $r$. Monte Carlo Simulations are used to validate the theoretical prediction of Theorem 1. Comparing the simulation results to the asymptotic MSE prediction of Theorem 1 shows the close match between the two. \figref{fig:ber_b} also shows the close match between simulation results and the asymptotic BER prediction of Theorem 2. We used $n=400, \beta =1.5, \alpha =0.5, T=1000, T_p =n$, and the data are averaged over $100$ independent Monte-Carlo iterations.
\begin{figure}
\begin{center}
\includegraphics[width = 9.8cm]{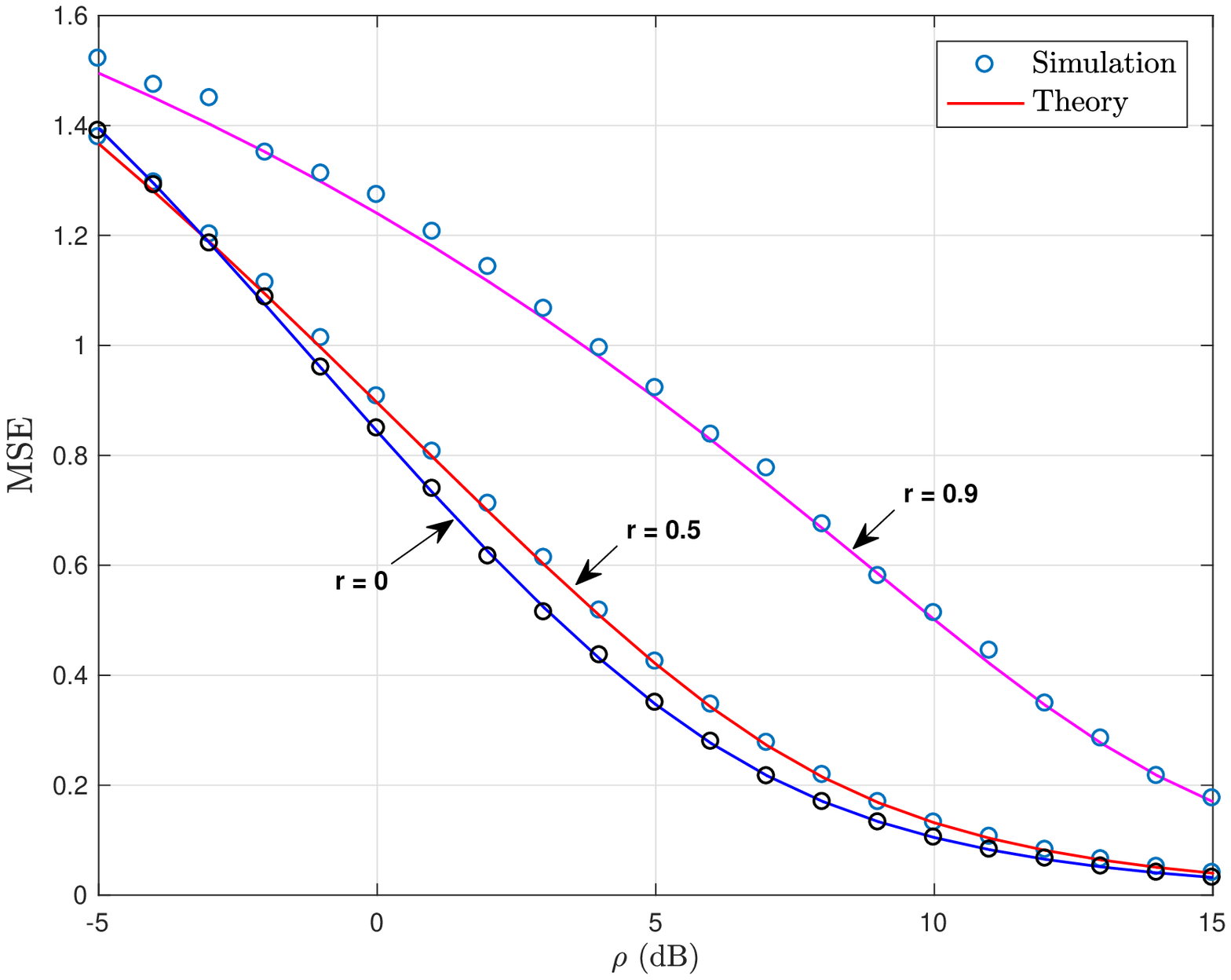}
\end{center}
\caption{\scriptsize {MSE performance of the BRO decoder.}}%
\label{fig:mse_b1}
\end{figure}
\begin{figure}
\begin{center}
\includegraphics[width = 9.8cm]{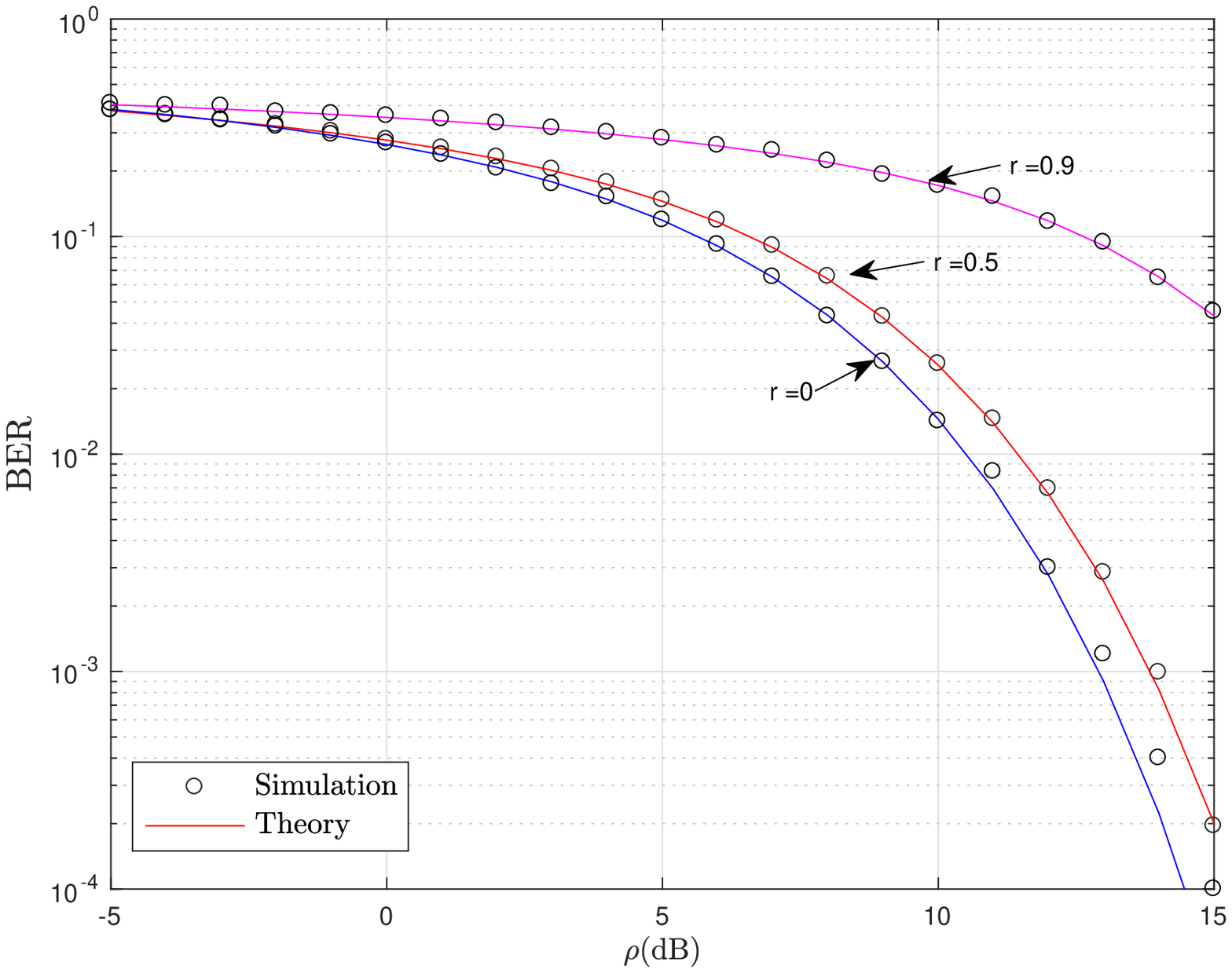}
\end{center}
\caption{\scriptsize {BER performance of the BRO decoder.}}%
\label{fig:ber_b}
\end{figure}

\noindent
\textbf{Double Descent:} In \figref{fig:double}, we plot the MSE and BER vs the ratio of the number of Rx antennas to the number of Tx antennas $\beta:=\frac{m}{n}$. From the figures, we can see that the BRO decoder clearly out performs the conventional Least Squares (LS) decoder (also known as the Zero-Forcing (ZF) decoder). 
Note that the MSE and BER of the \textbf{LS} detector first decrease for small values of $\beta$, then, it increases until it reaches a peak known as the interpolation threshold (at $\beta=1$) \cite{belkin2019reconciling}. After the peak, both metrics decrease monotonically as a function of $\beta$. This behavior is known as the \emph{double descent} phenomenon \cite{belkin2018understand, belkin2019reconciling}.
These figures show the important role played by the box-constraint of the BRO decoder in reducing MSE/BER and in the mitigation of the double descent phenomenon. This is expected since the box constraint of the BRO can be thought of as an $\ell_\infty$-norm regularizer. The authors of \cite{dhifallah2020precise, belkin2018understand,belkin2019reconciling} showed that optimal regularization can mitigate the double descent effect. 
\begin{figure*}
  \centering
\begin{subfigure}[h]{.5\textwidth}
\includegraphics[width =7.8cm]{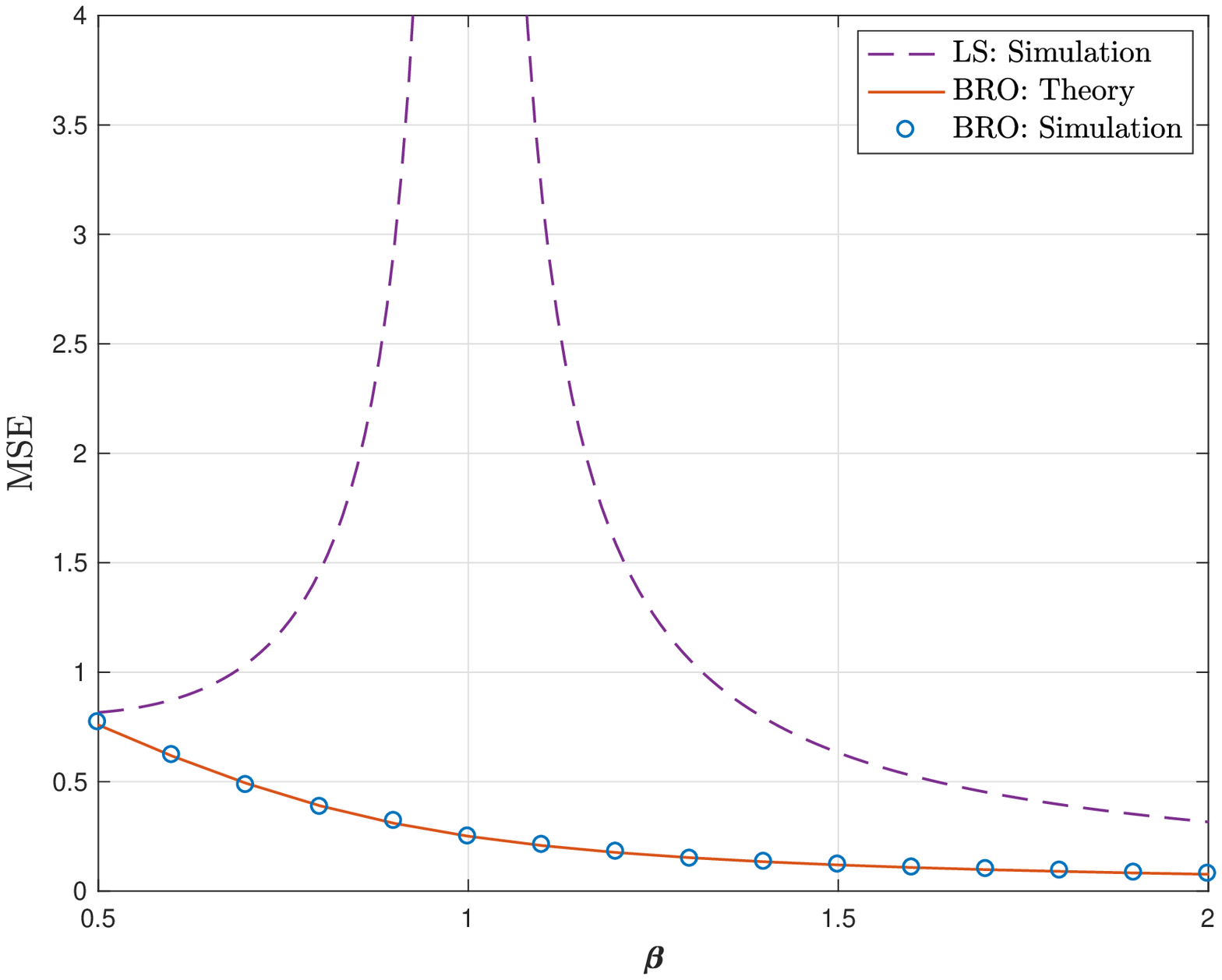}
\caption{\scriptsize {MSE performance.}}%
  \label{fig:sub-first}
\end{subfigure}%
\begin{subfigure}[h]{.5\textwidth}
\includegraphics[width =7.8cm]{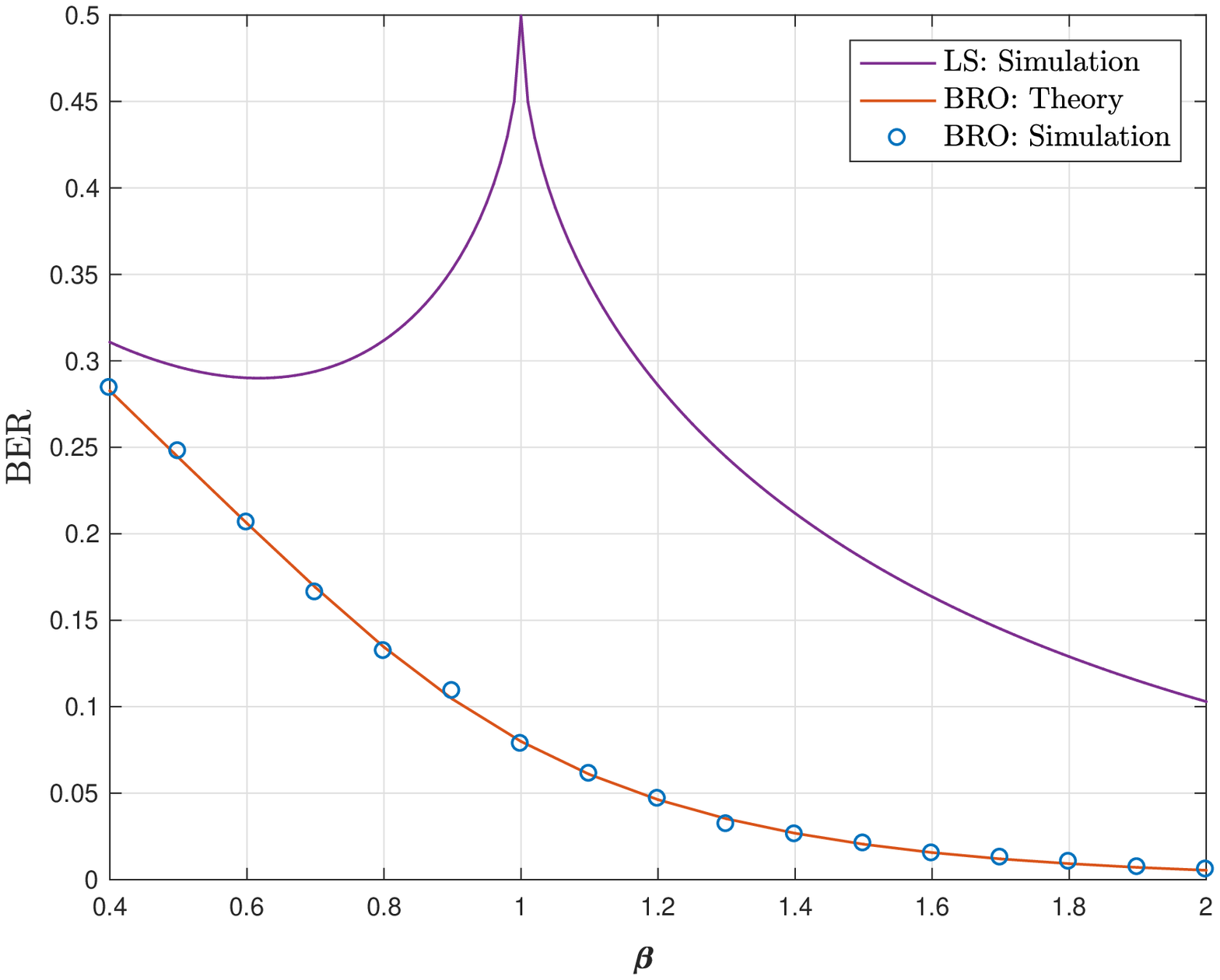}
  \label{fig:sub-second}
  \caption{\scriptsize {BER Performance.}}
\end{subfigure}
\caption{\scriptsize {Performance of BRO decoder vs $\beta$, with $n=400, r=0.2, \alpha =0.5, \rho=10 \mathrm {\ dB}, T=1000, T_p=n$.}}
\label{fig:double}
\end{figure*}
\section{Proof Outline}\label{sec:proof}
In this section, we use the CGMT framework \cite{thrampoulidis2018precise} to give a proof outline of Theorems 1 and 2. For the reader convenience, the CGMT is summarized next.
\subsection{Technical Tool: Convex Gaussian Min-max Theorem}
The CGMT allows us to replace the analysis of a generally hard optimization problem with a simplified auxiliary optimization problem. In this subsection, we recall the statement of the CGMT, and we refer the interested reader to \cite{thrampoulidis2018symbol, thrampoulidis2018precise} for the complete technical details.
In order to summarize the essential ideas, we consider the following two Gaussian processes:
\begin{align*}
X_{{\wv},{\uv}}&:= {\uv}^{\top}{\Gm \wv}+\psi({\wv},{\uv}),\\
Y_{{\wv},{\uv}}&:= \|{\wv}\|{\gv}^{\top}{\uv}+ \|{\uv}\|_2{\hv}^{\top}{\wv}+\psi({\wv},{\uv}),
\end{align*}
where $\Gm \in \mathbb{R}^{m \times n}, \boldsymbol{g} \in \mathbb{R}^{m}$, and $\hv \in \mathbb{R}^n$ are all assumed to have iid standard Gaussian entries. Further let $\mathcal{S}_w \subset \mathbb{R}^n $, and $\mathcal{S}_u \subset \mathbb{R}^m$ be convex and compact sets, and $\psi: \mathbb{R}^n \times \mathbb{R}^m \to \mathbb{R}$ is convex-concave continuous on $ \mathcal{S}_w \times \mathcal{S}_u$, possibly random but independent of $\Gm$. 
For these two processes, define the following random min-max optimization problems, which we refer to as the primary
optimization (PO) problem and the auxiliary optimization (AO):
\begin{subequations}
\begin{align}
\label{P,AO}
&\Phi^n(\Gm) := \underset{\wv \in \mathcal{S}_{w}}{\operatorname{\min}}  \ \underset{\uv \in \mathcal{S}_{u}}{\operatorname{\max}} \ X_{{\wv},{\uv}}, \\
&\phi^n(\gv, \hv) := \underset{\wv \in \mathcal{S}_{w}}{\operatorname{\min}}  \ \underset{\uv \in \mathcal{S}_{u}}{\operatorname{\max}} \ Y_{{\wv},{\uv}}.\label{AA2} 
\end{align}
\end{subequations}
Then, the CGMT relates the above problems as shown in the next theorem, the proof of which can be found in \cite{thrampoulidis2018precise}.
\begin{theorem}[CGMT \cite{thrampoulidis2018precise}]
Let $\mathcal{S}$ be any arbitrary open subset of $\mathcal{S}_w $, and $\mathcal{S}^c = \mathcal{S}_w \setminus\mathcal{S}$. Denote $\phi_{\mathcal{S}^c}^n(\gv,\hv)$ the optimal cost of the optimization in (\ref{AA2}) when the minimization over $\wv$ is constrained over $\wv \in \mathcal{S}^c$. Suppose that there exist constants $\overline{\phi}$ and $\overline{\phi}_{\mathcal{S}^c}$ such that (i) $\phi^n(\gv,\hv) \to \overline{\phi}$ in probability, (ii) $\phi_{\mathcal{S}^c}^n(\gv,\hv) \to \overline{\phi}_{\mathcal{S}^c}$ in probability, and (iii) $\overline{\phi} < \overline{\phi}_{\mathcal{S}^c}$.
Then, $\lim_{n \rightarrow \infty} \mathbb{P}[\wv_{\Phi} \in \mathcal{S}]= 1$, where $\wv_{\Phi}$ is a minimizer of (\ref{P,AO}).
\end{theorem}
It is not difficult to see that the conditions (i), (ii) and (iii) regarding the optimal cost of the AO imply that its solution $\wv_{\phi}$ satisfies: $\lim_{n \rightarrow \infty} \mathbb{P}[\wv_{\phi} \in \mathcal{S}]= 1$.
The non-trivial and powerful part of the CGMT is that the same conclusion holds true for the optimal solution
$\wv_\Phi$ of the PO \cite{thrampoulidis2018symbol}.

Having introduced the CGMT, we will proceed by providing an outline of the proof of Theorem 1 and Theorem 2.  The steps of the proof are presented in the following subsections.
\subsection{PO and AO Identification}
For notational convenience, we consider the error vector $\ev := \xv- \xv_0 $,
then the problem in \eqref{eq:BRO} can be reformulated as
\begin{equation}\label{Lasso_w}
\widehat{\ev} = \argmin_{-2\leq \ev\leq 0} \frac{1}{n} \left\| \sqrt{\frac{\rho_d}{n}} \widehat\Am \ev + \sqrt{\frac{\rho_d}{n}} \Deltam \xv_0 -\zv \right\|^2.
\end{equation}
Without loss of generality, we assume for the analysis that $\xv_0 =\mathbf{1}_n =[1,1, \cdots,1]^\top$. Then, ${\rm BER} = \frac{1}{n} \sum_{i=1}^n{\mathbbm{1}}_{\{\widehat{x}_i \leq 0 \}}$.
Next, note that $\widehat{\Am}$ can be written as 
$
\widehat{\Am} = \Rm_{\small{\widehat{A}}}^{1/2}\Bm,
$
with $\Bm$ being a Gaussian matrix with iid standard normal entries and $\Rm_{\widehat{\Am}}$ is the covariance matrix of $\widehat{\Am}$ as defined before. Thus, we have
\begin{equation}\label{eq:w2}
\widehat{\ev} = \argmin_{-2\leq \ev\leq 0} \frac{1}{n}\left\| \sqrt{\frac{\rho_d}{n}} \Rm_{\small{\widehat{A}}}^{1/2}\Bm \ev +\sqrt{\frac{\rho_d}{n}} \Deltam \xv_0-\zv \right\|^2,
\end{equation}
Since the Gaussian distribution is invariant under orthogonal transformations, and recalling that the spectral decomposition of $\Rm_{\widehat{\Am}}$ is $\Rm_{\widehat{\Am}} = \Um \Lambdam \Um^\top$, we have
\begin{equation}\label{eq:Gamma}
\widehat{\ev} =\argmin_{-2\leq \ev\leq 0} \frac{1}{n} \left\| \sqrt{\frac{\rho_d}{n}} \Lambdam^{1/2}\Bm \ev +\sqrt{\frac{\rho_d}{n}} \Deltam \xv_0-\zv \right\|^2.
\end{equation}
The loss function can be expressed in its dual form through the Fenchel's conjugate as 
\begin{align*}
 \left\| \sqrt{\frac{\rho_d}{n}} \left(\Lambdam^{1/2}\Bm \ev +\Deltam \xv_0\right)-\zv \right\|^2 =  \max_{\widetilde\uv}  \widetilde\uv^\top \left(\sqrt{\frac{\rho_d}{n}} \left(\Lambdam^{1/2}\Bm \ev +\Deltam \xv_0 \right)-\zv \right) -\frac{ \| \widetilde\uv \|^2}{4}.
\end{align*}
Then, \eqref{eq:Gamma} becomes
\footnote{One technical requirement of the CGMT is the compactness of the  feasibility set $\mathcal{S}_u$ which can be handled by noting that the optimal $\ev_*$ is bounded, and since $\widetilde\uv_* = 2\sqrt{\frac{\rho_d}{n}}(\Lambdam^{1/2} \Bm \ev_* +\Deltam \xv_0) - 2\zv$, then there exists a constant $C_{\widetilde{u}}>0$ such that $\| \widetilde\uv_* \| \leq C_{\widetilde u}$ with probability going to 1.}
\begin{align}\label{}
\mathfrak{D}_{1}^{n} : &\frac{1}{n} \min_{-2\leq \ev\leq 0} \max_{\widetilde\uv} \! \sqrt{\frac{\rho_d}{n}} \widetilde\uv^\top \Lambdam^{1/2}\Bm \ev +\sqrt{\frac{\rho_d}{n}}  \widetilde\uv^\top \!\Deltam \xv_0-\widetilde\uv^T\zv\! -\frac{ \| \widetilde\uv \|^2}{4}.
\end{align}
Defining $\uv := \frac{1}{\sqrt{n}}\Lambdam^{1/2} \widetilde\uv$ yields
\begin{align}\label{PO_L}
\mathfrak{D}_{2}^{n} :  \min_{-2\leq \ev\leq 0} \max_{\uv \in \mathcal{S}_u} & \frac{\sqrt{\rho_d} }{n}\uv^\top  \Bm \ev +\frac{\sqrt{\rho_d} }{n} \uv^\top \Lambdam ^{-1/2} \Deltam \xv_0  - \frac{1}{\sqrt n} \uv^\top \Lambdam^{-1/2}\zv  -\frac{1 }{4} \uv^\top \Lambdam^{-1} \uv,
\end{align}
where $\mathcal{S}_u = \{ \uv \in \mathbb{R}^m : \| \uv \| \leq C_u \}$, for some fixed constant $C_u>0$ that is independent of $n$.
The above optimization is in a PO form, and its \emph{corresponding AO} is
\begin{align}\label{AO_L}
\widetilde{\mathfrak{D}}_{1}^{n} :& \min_{-2\leq \ev\leq 0} \max_{\uv \in \mathcal{S}_u}  \frac{\sqrt {\rho_d}}{n} \| \ev \| \gv^\top \uv -\frac{\sqrt {\rho_d}}{n} \| \uv \| \hv^\top \ev  +\frac{\sqrt{\rho_d} }{n} \uv^\top \Lambdam ^{-1/2} \Deltam \xv_0  - \frac{1}{\sqrt{n}} \uv^\top \Lambdam^{-1/2}\zv -\frac{1}{4} \uv^\top \Lambdam^{-1} \uv,
\end{align}
where $\gv \sim \mathcal{N}(\mathbf{0}, \Id_m)$ and $\hv \sim \mathcal{N}(\mathbf{0}, \Id_n)$ are independent vectors.

Fixing the norm of the normalized error vector $\frac{\ev}{\sqrt{n}}$ to $\xi := \frac{\| \ev \|}{\sqrt{n}}$,
we get
\begin{align}\label{AO_11}
&\widetilde{\mathfrak{D}}_{2}^{n} : \min_{\xi\geq 0} \max_{\uv \in \mathcal{S}_u}  \! \xi \sqrt{\frac{\rho_d}{n}}\gv^\top \uv \!+\!\frac{\sqrt{\rho_d} }{n} \uv^\top \Lambdam ^{-1/2}\! \Deltam \xv_0  \nonumber \\
  &-\! \frac{1}{\sqrt{n}} \uv^\top \Lambdam^{-1/2}\!\zv \!-\frac{1 }{4} \uv^\top \Lambdam^{-1} \uv + \!{\sqrt{\rho_d} \| \uv\|} \!\min_{\substack{\| \ev \| =\sqrt{n} \xi\\ -2\leq \ev\leq 0}}\! -\frac{1}{n}  \hv^\top \ev.
\end{align}
Now, applying Lemma \ref{lemma1}, and under the condition, $\xi^2\leq \frac{4}{n} \sum_{i=1}^n\mathbbm{1}_{\{h_i< 0\}}$, we have
\begin{align}\label{eq:AO}
\widehat \Upsilon(\xi)&:=\min_{\substack{\| \ev \| =\sqrt{n} \xi\\ -2\leq \ev\leq 0}}  -\frac{1}{n}  \hv^\top \ev\\ &=\!\frac{2}{n}\!\left[\sum_{i=1}^n\!|h_i|\mathbbm{1}_{\{h_i\leq-\hat\mu(\xi)\}}\!+\!\sum_{i=1}^n\!\frac{|h_i|^2}{\hat\mu(\xi)}\!\mathbbm{1}_{\{-\hat\mu(\xi)<h_i<0\}}\right]\!,
\end{align}
where $\hat\mu(\xi)$ verifies
\begin{align}\label{eq94}
\frac{4}{n}\sum_{i=1}^n \mathbbm{1}_{\{h_i\leq-\hat\mu(\xi)\}}+\frac{4}{\hat\mu(\xi)^2n}\sum_{i=1}^n |h_i|^2\mathbbm{1}_{\{-\hat\mu(\xi)<h_i<0\}}=\xi^2.
\end{align}
In the limit as $n\to \infty$, equation \eqref{eq94} converges to
\begin{align}\label{eq95}
4\left(Q(\mu(\xi))+\frac{\varphi(\mu(\xi))}{\mu^2(\xi)} \right)=\xi^2,
\end{align}
with $\varphi(t):=\frac{1}{2}-Q(t)-\frac{x}{\sqrt{2\pi}}e^{-t^2/2}$, and the following convergence results hold true
\begin{subequations}
\begin{align}\label{}
&  \widehat \Upsilon(\xi)-\Upsilon(\mu(\xi))\pto 0,\\
&  \hat \mu(\xi)-{\mu}(\xi)\pto 0,
\end{align}
\end{subequations}
with $\Upsilon(t):=\frac{1}{t}\left(1-2Q(t)\right)$ and
$ \mu(\xi)$ is the unique positive solution of the following fixed-point equation
\begin{equation}\label{eq:tau}
 \mu(\xi)=2\sqrt{\frac{\varphi( \mu(\xi))}{\xi^2-4Q( \mu(\xi))}},
\end{equation}
where \eqref{eq:tau} is found from \eqref{eq95}. For notational simplicity, we will denote $\mu(\xi)$ by $\mu$ only.
Note that $\frac{1}{n} \sum_{i=1}^n\mathbbm{1}_{\{h_i< 0\}}\pto \frac{1}{2} $, then the condition  $\xi^2\leq \frac{4}{n} \sum_{i=1}^n\mathbbm{1}_{\{h_i< 0\}}$ asymptotically becomes  $\xi^2\leq 2$. Then, applying Lemma 10 in \cite{thrampoulidis2018precise}, we have $\widetilde{\mathfrak{D}}_{2}^{n} -\widetilde{\mathfrak{D}}_{3}^{n} \pto 0$, where
\begin{align}\label{}
\widetilde{\mathfrak{D}}_{3}^{n}:& \min_{0<\xi \leq \sqrt{2}} \max_{\uv\in \mathcal{S}_u}  \ \xi \sqrt{\frac{\rho_d}{n}}\gv^\top \uv +\frac{\sqrt{\rho_d} }{n} \uv^\top \Lambdam ^{-1/2} \Deltam \xv_0  \nonumber \\
  & - \frac{1}{\sqrt{n}} \uv^\top \Lambdam^{-1/2}\zv   -\frac{1 }{4} \uv^\top \Lambdam^{-1} \uv - {\sqrt{\rho_d} \| \uv\|}  \Upsilon({\mu}).
\end{align}
From the identity, $ \| \uv \| =\min_{\chi>0} \frac{\chi}{2} + \frac{\| \uv \|^2}{2 \chi}$, we can write
\begin{align}\label{}
\widetilde{\mathfrak{D}}_{4}^{n}:& \!\min_{0<\xi\leq \sqrt{2}} \max_{\substack{\uv \in \mathcal{S}_u \\ \chi>0}} \frac{1}{\sqrt{n}} \!\bigg(\! \sqrt{\rho_d}\xi \gv\! +\! \Lambdam^{-1/2}\! \left(\sqrt{\frac{\rho_d}{n}} \!\Deltam \xv_0 -\zv\!\right)\!\bigg)^\top \!\uv \nonumber \\
  & - \frac{ \sqrt{\rho_d}\chi}{2 } \Upsilon(\mu)  - \frac{\sqrt{\rho_d} \Upsilon(\mu)\|\uv\|^2}{2 \chi }  -\frac{1}{4} \uv^\top \Lambdam^{-1} \uv.
\end{align}
Let $\widetilde{\gv}:= \sqrt{\rho_d}\xi \gv + \Lambdam^{-1/2} \left(\sqrt{\frac{\rho_d}{n}} \Deltam \xv_0 -\zv\right)$, and $\Xim :=  \frac{1 }{2}\Lambdam^{-1} +\frac{\sqrt{\rho_d} }{ \chi } \Upsilon(\mu) \Id_{m}$, then
\begin{align}\label{}
\widetilde{\mathfrak{D}}_{5}^{n} :& \min_{0<\xi\leq \sqrt{2}} \max_{\substack{\uv\in \mathcal{S}_u\\\chi>0}} \ \frac{1}{\sqrt{n}}\widetilde\gv^\top \uv
-\frac{1 }{2} \uv^\top \Xim \uv- \frac{ \chi \sqrt{\rho_d}}{2 } \Upsilon(\mu).
\end{align}
The optimal $\uv_*$ can be easily found as 
 $
\uv_*=\frac{1}{\sqrt{n}}\Xim^{-1} \widetilde\gv
$.
Thus, the AO can be written as
\begin{align}\label{}
\widetilde{\mathfrak{D}}_{6}^{n}: \min_{0<\xi\leq \sqrt{2}} \max_{ \chi>0} & \ \frac{1}{2n}\widetilde{\gv}^\top \Xim^{-1} \widetilde{\gv}^\top - \frac{\sqrt{\rho_d} \chi}{2} \Upsilon(\mu).
\end{align}
\subsection{Large System Analysis of the AO}
Note that $\widetilde \gv$ is distributed as $\mathcal{N}(\boldsymbol{0},\Rm_{\widetilde\gv})$, with covariance matrix $\Rm_{\widetilde\gv} := \mathbb{E}[\widetilde{\gv} \widetilde{\gv}^\top]$ that is given by
$$
\Rm_{\widetilde\gv} = \rho_d \xi^2\bI_m+ \rho_d \Rm_{\Delta} \Lambdam^{-1}+ \Lambdam^{-1}.
$$
Then, using the trace lemma \cite{Couillet2011}, we get
$ \frac{1}{n}\widetilde{\gv}^\top \Xim^{-1} \widetilde{\gv}^\top - \frac{1}{n}\tr\left(\Rm_{\widetilde\bg}\Xim^{-1}\right)\pto 0.$
Again using Lemma 10 of \cite{thrampoulidis2018precise}, $\widetilde{\mathfrak{D}}_{6}^{n}- \widetilde{\mathfrak{D}}_{7}^{n} \pto 0,$ where
\begin{equation}
\widetilde{\mathfrak{D}}_{7}^{n} \!: \!\min_{0<\xi\leq \sqrt{2}} \!\max_{\chi >0} \!\frac{1}{2 n} \!\sum_{j=1}^{m} \frac{\rho_d \lambda_{j} \xi^2 \!+\! \rho_d [\Rm_\Delta]_{jj} +1}{\frac{1}{2} + \frac{\sqrt{\rho_d} \lambda_j \Upsilon(\mu)}{ \chi}} \!-\!\frac{\sqrt{\rho_d}}{2} \!\Upsilon(\mu)  \chi.
\end{equation}
Performing the change of variable $\gamma :=\frac{\chi}{\Upsilon(\mu)}$, we have
\begin{align}
\label{alp_obj}
\widetilde{\mathfrak{D}}_{8}^{n} : &\min_{0<\xi\leq \sqrt{2}} \max_{\gamma >0} \frac{1}{2 n} \!\sum_{j=1}^{m} \frac{\rho_d \lambda_{j} \xi^2 \!+\! \rho_d [\Rm_\Delta]_{jj} +1}{\frac{1}{2} + \frac{\sqrt{\rho_d} \lambda_j }{ \gamma}} -\frac{\sqrt{\rho_d}}{2} \Upsilon^2(\mu) \gamma.
\end{align}
Using equation \eqref{eq95}, we have
\begin{equation}\label{eq:psi}
\xi^2 =F(\mu):= 4 \bigg( Q(\mu)+ \frac{f(\mu)}{\mu^2} \bigg).
\end{equation}
One can easily show that $F(\cdot)$ is a strictly decreasing function on $(0,\infty)$, and then using the change of variables rule in \cite[page 130]{boyd}, we can make the change of variable $\xi^2 = F(\mu)$ to get 
\begin{align}
\label{obj_func}
\widetilde{\mathfrak{D}}_{9}^{n} : \min_{\mu>0} \max_{\gamma >0} \frac{1}{2 n} \!\sum_{j=1}^{m}& \frac{\rho_d \lambda_{j} F(\mu) \!+\! \rho_d [\Rm_\Delta]_{jj} +1}{\frac{1}{2} + \frac{\sqrt{\rho_d} \lambda_j }{ \gamma}} 
-\frac{\sqrt{\rho_d}}{2} \Upsilon^2(\mu) \gamma.
\end{align}
One can show that the cost function of $\widetilde{\mathfrak{D}}_{9}^{n} $ is strictly positive for all $\gamma>0$ by checking its second derivative with respect to $\mu$. Hence, $\widetilde{\mathfrak{D}}_{9}^{n}$ has a unique minimizer $\mu_*$. This implies that $\widetilde{\mathfrak{D}}_{8}^{n}$ has a unique minimizer $\xi_*$.
\subsection{Using the CGMT: Theorm~1 and Theorem~2 Proofs}
We start by proving Theorem~1, where we analyze the asymptotic behavior of the MSE of the BRO. 
Let $\widetilde \ev$ be the optimal solution to the AO defined as the solution to $\widetilde{\mathfrak{D}}_{1}^{n}$. 
Define $\hat\xi$ as the minimizer of \eqref{AO_11}. By definition, ${\hat{\xi}}^2= \frac{ \| \widetilde\ev \|^2}{n}$. In the previous section, we have shown that $\widetilde{\mathfrak{D}}_{1}^{n}-\widetilde{\mathfrak{D}}_{8}^{n}\pto 0$, and since $\widetilde{\mathfrak{D}}_{8}^{n}$ in \eqref{alp_obj} has a unique minimizer $\xi_*$, then $\hat\xi-\xi_*\pto 0$ which implies that
for any $\varepsilon>0$, and with probability approaching 1 (w.p.a.1), we have
$$
\widetilde\ev  \in \mathcal{S}_{\rm{MSE}} := \bigg\{ \sv \in \mathbb{R}^n: \bigg| \frac{1}{n} \| \sv \|^2 - F( \mu_*) \biggr| < \varepsilon \bigg\},
$$
where $F( \mu)$ is as defined in \eqref{eq:psi} and $\mu_*$ be the optimal solution to $\widetilde{\mathfrak{D}}_{9}^{n}$. 
Then, applying the CGMT yields that $\widehat\ev \in \mathcal{S}_{\rm{MSE}}$ w.p.a.1 as well. This ends the proof of Theorem 1.

For the BER analysis, we start by changing the set $\mathcal{S}_{\rm{MSE}}$ to the following:
$$
\mathcal{S}_{\rm{BER}} := \bigg\{\sv \in \mathbb{R}^n: \bigg| \frac{1}{n} \sum_{i=1}^{n} \mathbbm{1}_{\{{s}_i \leq -1 \} } - Q \left(\frac{\mu_*}{2} \right) \bigg| < \varepsilon \bigg\}.
$$
Using the expression of $\widetilde\ev$ in \eqref{w_star}, we can show that
$$
\frac{1}{n}\sum_{i=1}^n\mathbbm{1}_{\{\widetilde e_i\leq -1\}} =
 \frac{1}{n}\sum_{i=1}^n \mathbbm{1}_{\{h_i\leq -\hat\mu\}}+\frac{1}{ n}\sum_{i=1}^n \mathbbm{1}_{\{-\hat\mu<h_i\leq -\frac{\hat\mu}{2}
 \}}.
 $$
Then, the right hand side (RHS) of the above equation converges as
$
\left| {\rm{RHS}}- Q\left(\frac{\mu_*}{2}\right)\right|\pto 0.
$
Therefore, $\widetilde\ev \in \mathcal{S}_{\rm{BER}} $ w.p.a.1.\footnote{The indicator function $\mathbbm{1}_{\{\widetilde e_i\leq -1\}} $ is not Lipschitz, so the CGMT cannot be directly applied. However, as discussed in \cite[Lemma A.4]{thrampoulidis2018symbol}, this function can be appropriately approximated with Lipschitz functions.} 
Then, by the CGMT, $\widehat{\ev} \in \mathcal{S}_{\rm{BER}}$ with probability approaching 1, thus concluding the proof of Theorem \ref{thmmm2}.
\section{Conclusion}\label{sec:conclusion}
In this work, we derived precise characterization of the asymptotic behavior of the BRO decoder under the assumptions of imperfect CSI and receiver-side correlation. We used the MSE and BER as performance metrics of the decoder. Then, we derived the optimal power allocation between pilot and data symbols using the presented asymptotic results. Simulation results validate our theoretical analysis even for small dimensions of the problem. They also show that the BRO decoder can mitigate the double descent effect encountered with other conventional decoders such as the ZF decoder. For simplicity of the analysis, BPSK signals are used but our results can be extended to higher order modulation schemes such as $M$-PAM and is left for future work. Possible future work includes studying the fully correlated MIMO systems (Kronecker correlation), and analyzing the regularized version of the BRO decoder.
\section*{Acknowledgment}
This work was supported by the University of Ha'il, Saudi Arabia.
\begin{appendices}
\section*{Appendix}
\begin{lemma}[\cite{alrashdi2020box}] Let $\bh \in \mathbb{R}^{n}$ and let $a$ and $\xi$ be strictly positive constants such as \begin{align}\label{cond}
\xi^2\leq a^2 \left(\frac{1}{n}\sum_{i=1}^n \mathbbm{1}_{\{h_i< 0\}}\right).
\end{align}
Then, 
\begin{align*}
&\min_{\substack{\| \ev \| =\sqrt{n} \xi\\ -a \leq \ev\leq 0}}  -\frac{1}{n}  \hv^\top \ev=\\ &-\frac{a}{n}\sum_{i=1}^n |h_i|\mathbbm{1}_{\{h_i\leq-\mu\}}-\frac{a}{\mu n}\sum_{i=1}^n |h_i|^2\mathbbm{1}_{\{-\mu<h_i<0\}},
\end{align*}
where $\mu$ satisfies
\begin{align*}
\frac{a^2}{n}\left(\sum_{i=1}^n \mathbbm{1}_{\{h_i\leq-\mu\}}+\frac{1}{\mu^2}\sum_{i=1}^n |h_i|^2\mathbbm{1}_{\{-\mu<h_i<0\}} \right) =\xi^2.
\end{align*}
The corresponding optimal $\ev^*$ is given by
\begin{equation}\label{w_star}
e_i^*=\begin{cases}0, \ \ \ \ \ {\rm if}\ \  h_i\geq 0\\
\frac{a}{\mu} h_i,\ \  {\rm if}\ \ -\mu<h_i<0\\
-a,   \ \ \  {\rm if}\ \ h_i\leq-\mu. \end{cases}
\end{equation}
\label{lemma1}
\end{lemma}
\begin{proof}
See \cite{alrashdi2020box}.
\end{proof}
\end{appendices}
\bibliographystyle{IEEEbib}
\bibliography{References.bib}
\end{document}

%% file: macros.tex
\long\def\comment#1{}

\DeclareMathOperator*{\argmax}{arg\,max}
\DeclareMathOperator*{\argmin}{arg\,min}

\newfont{\bbb}{msbm10 scaled 700}

\newfont{\bb}{msbm10 scaled 1100}

\newcommand{\ev}{{\boldsymbol e}}

\newcommand{\gv}{{\boldsymbol g}}
\newcommand{\hv}{{\boldsymbol h}}

\newcommand{\sv}{{\boldsymbol s}}

\newcommand{\uv}{{\boldsymbol u}}
\newcommand{\wv}{{\boldsymbol w}}

\newcommand{\xv}{{\boldsymbol x}}
\newcommand{\yv}{{\boldsymbol y}}
\newcommand{\zv}{{\boldsymbol z}}

\newcommand{\Am}{{\boldsymbol A}}
\newcommand{\Bm}{{\boldsymbol B}}

\newcommand{\Gm}{{\boldsymbol G}}
\newcommand{\Hm}{{\boldsymbol H}}
\newcommand{\Id}{{\boldsymbol I}}

\newcommand{\Rm}{{\boldsymbol R}}

\newcommand{\Um}{{\boldsymbol U}}

\newcommand{\Xm}{{\boldsymbol X}}
\newcommand{\Ym}{{\boldsymbol Y}}
\newcommand{\Zm}{{\boldsymbol Z}}

\newcommand{\Lambdam}{\hbox{\boldmath$\Lambda$}}
\newcommand{\Deltam}{\hbox{\boldmath$\Delta$}}

\newcommand{\Xim}{\hbox{\boldmath$\Xi$}}

